\documentclass[oneside]{amsart}
\usepackage{geometry}
\pdfoutput=1

\usepackage{graphicx}
\usepackage{hyperref}

\usepackage[T1]{fontenc}
\usepackage[utf8]{inputenc}
\usepackage{lmodern}
\usepackage{amssymb,amsmath}

\newtheorem{thm}{Theorem}[section]
\newtheorem{prop}[thm]{Proposition}

\newtheorem{cor}[thm]{Corollary}
\newtheorem{definition}[thm]{Definition}
\newtheorem{remark}[thm]{Remark}
\theoremstyle{definition}
\newtheorem{example}[thm]{Example}

\usepackage{xcolor}

\definecolor{col1}{RGB}{100,143,255}
\definecolor{col2}{RGB}{120, 94, 240}
\definecolor{col3}{RGB}{254,97,0}
\definecolor{col4}{RGB}{220, 38, 127}
\definecolor{col5}{RGB}{255, 176, 0}

\hypersetup{
  pdfauthor={Michael Borinsky, Gerald V. Dunne, Karen Yeats},
  pdftitle={Tree-tubings and the combinatorics of resurgent Dyson-Schwinger equations},
  pdfsubject={},
  pdfkeywords={},
  linkcolor  = col2,
  citecolor  = col4,
  urlcolor   = col5,
  colorlinks = true,
}

\newcommand{\bigO}{\mathcal{O}}
\newcommand{\Z}{{\mathbb{Z}}}
\newcommand{\Q}{{\mathbb{Q}}}
\newcommand{\R}{{\mathbb{R}}}

\newcommand{\A}{{\mathcal{A}}}
\newcommand{\dd}{{\mathrm{d}}}

\title[Tree-tubings and the combinatorics of resurgent Dyson--Schwinger equations]{Tree-tubings and the combinatorics of\\ resurgent Dyson--Schwinger equations}

\author{Michael Borinsky \and Gerald V. Dunne \and Karen Yeats}

\address{
Michael Borinsky\\
Perimeter Institute for Theoretical Physics\\
31 Caroline St N\\
Waterloo\\
ON N2L 2Y5\\
Canada
}
\address{
Michael Borinsky\\
Institute for Theoretical Studies\\
ETH Z\"urich\\
Clausiusstrasse 47\\
8006 Zürich\\
Switzerland\\
}

\address{
Gerald V. Dunne\\
Department of Physics\\ 
University of Connecticut\\
   196 Auditorium Road\\
   Storrs CT 06269-3046\\
USA
}

\address{
Karen Yeats\\
Department of Combinatorics and Optimization\\
University of Waterloo\\
200 University Ave. W.
Waterloo ON, N2L 3G1
Canada
}

\begin{document}

\begin{abstract}
    We give a novel combinatorial interpretation to the perturbative series solutions for a class of Dyson-Schwinger equations.  We show how binary tubings of rooted trees with labels from an alphabet on the tubes, and where the labels satisfy certain compatibility constraints, can be used to give series solutions to Dyson-Schwinger equations with a single Mellin transform which is the reciprocal of a polynomial with rational roots, in a fully combinatorial way.  Further, the structure of these tubings leads directly to systems of differential equations for the anomalous dimension that are ideally suited for resurgent analysis.  We give a general result in the distinct root case, and investigate the effect of repeated roots, which drastically changes the asymptotics and the transseries structure.
\end{abstract}

\maketitle
\section{Introduction}

Dyson-Schwinger equations are tools to study non-perturbative aspects of quantum field theories. A certain type of Dyson-Schwinger equations has the basic structure,
\begin{equation}\label{DSE k=1 s=-2}
G(x,L) = 1 + x\, G(x, \partial_\rho)^{-1} (e^{L\rho}-1)F(\rho)|_{\rho=0}
\end{equation}
where $F(\rho)$ is a given formal Laurent series in $\rho$ with a simple pole and $\partial_\rho$ is the usual differential operator associated to $\rho$.
In the ring of formal power series \eqref{DSE k=1 s=-2} uniquely defines a bivariate formal power series $G(x,L) $ whose coefficients are polynomials in the coefficients of $F(\rho)$. In \cite{BCEFNOHY}, the third author, with other coauthors, proved that the power series solution to \eqref{DSE k=1 s=-2}, as well as to more general equations of a similar nature, can be viewed purely combinatorially and constructed explicitly in terms of tubings of rooted trees, the definition of which will be given below. Consequently, solutions of Dyson-Schwinger equations of this form can be constructed using relatively standard tools of combinatorial generating functions.

Here, we study the solutions of \eqref{DSE k=1 s=-2}, when $F(\rho)$ is the reciprocal of a polynomial with rational roots.
Our main results are that, in this case, the coefficients of power series solutions of \eqref{DSE k=1 s=-2} 
count certain labelled tubings of rooted trees and that these solutions can be
described systematically by univariate nonlinear differential equation systems (Theorem~\ref{thm comb de}).

From a physics perspective, the Dyson-Schwinger equations in \eqref{DSE k=1 s=-2} arise when approximating a correlation function's full perturbative expansion in a specific quantum field theory by neglecting all Feynman diagram contributions that lie outside a particular family of diagrams. In the present case, the contributing family of diagrams is constructed by inserting a single 1-loop propagator graph into itself, where the insertion is into only one internal propagator, but in that propagator, the recursive insertions can be nested and chained. The series $F(\rho)$ encodes the \emph{Mellin transform} of the integral for the 1-loop graph with the edge into which we are inserting regularized by $\rho$. For that reason, we refer to the Laurent series $F(\rho)$ as `the Mellin transform' through out the paper. In Section~\ref{sec:approxs}, we discuss a series of more invasive approximations that are often used in the physics literature and explain how our formulas directly degenerate to these further approximations.

The solutions to Dyson-Schwinger equations \cite{bk30,bkerfc} generate interesting resurgent transseries, as was shown by the first two authors with other coauthors~\cite{BD,BDM,BB}. %
Our specific set of solutions to \eqref{DSE k=1 s=-2} incorporates all previously studied 
examples as special cases. It thereby provides a combinatorial perspective previously lacking in all but the simplest of these examples \cite{BD}. %

Dyson-Schwinger equations are one of the few tools with which the divergence of the perturbative expansions in quantum field theories can be studied. 
As expected by physically motivated
arguments in quantum field theory \cite{Dyson:1952tj}, the coefficients of our solutions of \eqref{DSE k=1 s=-2} often grow factorially.
In Section~\ref{sec:asymptotics}, we study this asymptotic growth rate using the \emph{ring of factorially divergent power series}, developed by the first author \cite{Bgenfun}, a simplified version of \emph{alien calculus}, which, in its fully fledged form, is one of the pillars of the \emph{resurgence} framework \cite{ecalle,MR3526111}. Additionally, we illustrate the \emph{transseries} approach to study the asymptotic properties of the solutions to \eqref{DSE k=1 s=-2}. %

We will illustrate our results using the special, previously studied cases.
The two cases examined in~\cite{bkerfc} can both be written in the form of \eqref{DSE k=1 s=-2}, 
with $F(\rho) = 1/(\rho(2-\rho))$ in their Yukawa case and 
$F(\rho) = 1/(\rho(1-\rho)(2-\rho)(3-\rho))$ in their $\phi^3$ case
(see Remark~\ref{rem:rescaling} for differing prefactor and scaling conventions for $F(\rho)$ and for $G(x,L)$).  Resurgent analysis of these two cases was accomplished in \cite{BD} and \cite{BDM}, respectively, with further work on the $\phi^3$ case along with an intermediate model in \cite{BB}.  In the latter paper, the intermediate model was somewhat ad-hoc, but in our context, it also appears from the same Dyson-Schwinger equation with Mellin transform $F(\rho) = 1/(\rho(1-\rho)(2-\rho))$. In this way, the intermediate model appears more naturally in the present setup; it is simply a different integral for the 1-loop bubble.

\section*{Acknowledgements}

KY thanks Nick Olson-Harris, and GD thanks Ovidiu Costin, for useful discussions. 
We also thank Paul Balduf, Marc Bellon, David Broadhurst, and Dirk Kreimer for their helpful comments.

The authors are grateful to the Isaac Newton Institute for Mathematical Sciences in Cambridge, for support and hospitality during the programme ``Applicable resurgent asymptotics: towards a universal theory,'' where this project was begun, to Perimeter Institute for some of the middle of the project,
and to the Max Planck Institute for Mathematics in Bonn for its hospitality and financial support when completing this project.  Research at Perimeter Institute is supported in part by the Government of Canada through the
Department of Innovation, Science and Economic Development Canada and by the province of
Ontario through the Ministry of Economic Development, Job Creation and Trade. 

MB is supported by Dr.\ Max Rössler, the Walter Haefner Foundation, and the ETH Zürich Foundation.
GD is supported by the U.S.\ Department of Energy, Office of High Energy Physics, Award DE-SC0010339.
KY is supported by an NSERC Discovery grant and the Canada Research Chairs program, and through part of this work by the Simons Foundation via the Emmy Noether Fellows Program at Perimeter Institute. 

\section{Solving Dyson-Schwinger equations by tubings}

To explain the solution of \eqref{DSE k=1 s=-2} in terms of tubings of plane rooted trees, we need a few definitions.  A \emph{plane rooted tree}, defined recursively, is a vertex $r$, called the root, and a possibly empty ordered list of plane rooted trees whose roots are the children of $r$.  A plane rooted tree can also be seen as a tree in the graph theory sense with a special vertex $r$, called the root, and at each vertex $v$ a linear order on the edges incident to $v$ which are not on the path from $v$ to $r$.

A \emph{tube} of a plane rooted tree $t$ is a set of vertices of $t$ that induces a connected subgraph of $t$.  Note that each tube inherits a plane rooted tree structure from $t$.  A \emph{tubing} $\tau$ of $t$ is a set of tubes of $t$ including the \emph{outer tube} consisting of all vertices of $t$ and for which for every tube $a$ containing at least two vertices in $\tau$ there exist two other tubes of $\tau$ which partition $a$.  These are the tubings studied in \cite{BCEFNOHY}, and they can also be called \emph{binary tubings} to distinguish them from other notions of tubing in the literature, see the appendix of \cite{BCEFNOHY}.

The number of tubes in a tubing $\tau$ of a tree $t$ with $n$ vertices is $2n-1$. This can be proved by a quick induction on the binary structure of the tubing.  Another useful thing to note is that for a fixed tubing $\tau$ every edge of the graph corresponds to a partitioning of some tube into two subtubes in $\tau$, and every such partitioning in $\tau$ corresponds to an edge of the tree.

The logarithmic divergence of the Mellin transformed Feynman diagram implies that $F(\rho)$ has a simple pole at the origin.  So, in accordance to the physical origins of \eqref{DSE k=1 s=-2}, we will assume that $F(\rho)$ can be written as the following Laurent series expansion
\[
F(\rho) = \sum_{i\geq 0} c_i\rho^{i-1}.
\]
\begin{remark}
Symmetries of specific quantum field theories can further constrain 
the structure of the Laurent series $F(\rho)$. 
For instance, the conformal symmetry of Yukawa theory implies that $F(\rho) = F(2-\rho)$.
Here, we only assume that $F(\rho)$ is of the general form above.
\end{remark}

Given a tubing $\tau$ of a plane rooted tree $t$, define the \emph{Mellin monomial} of $\tau$ to be
\[
c(\tau) = \prod_{\substack{v\in V(t) \\  v \text{ not the root of $t$}}} c_{b(v,\tau)-1}.
\]
where $b(v, \tau)$ is the number of tubes of $\tau$ for which $v$ is the root of that tube.

Then, the main result of \cite{BCEFNOHY} when restricted to the special case of \eqref{DSE k=1 s=-2} is that the contribution of a tubing $\tau$ to the power series solution of \eqref{DSE k=1 s=-2} is
\[
\phi_L(\tau) = c(\tau)\sum_{i=1}^{b(\tau)}c_{b(\tau)-i} \frac{L^i}{i!}.
\]
where $b(\tau) = b(r, \tau)$ for $r$ the root of the tree, and so the power series solution of \eqref{DSE k=1 s=-2} is
\begin{equation}\label{eq tubing soln}
G(x,L) = \sum_{\substack{t \text{ plane }\\\text{rooted tree}}} x^{|t|} \sum_{\tau \text{ tubing of }t}\phi_L(\tau),
\end{equation}
where $|t|$ is the number of vertices of $t$.

Power series solutions (called perturbative solutions in the physics context) to more general Dyson-Schwinger equations are similar but include a product of binomial factors for each tree and weights on the vertices for multiple Mellin transforms, see \cite{BCEFNOHY} for details.

The coefficient of $L^1$ in $G(x,L)$ is called the \emph{anomalous dimension}, $\gamma(x)$, an important function for physical applications.  The rest of $G(x,L)$ can be recovered from $\gamma(x)$ using the renormalization group equation, which in this context (see \cite{Ymem}) tells us that
\[
\Gamma_k(x) = \frac{1}{k} \gamma(x) \left( 1-2x \partial_x \right) \Gamma_{k-1}(x) 
\]
where $G(x,L) = 1+\sum_{k= 0}^\infty \Gamma_k(x)L^k$, and so $\Gamma_1(x) = \gamma(x)$.
Note that the anomalous dimension is particularly nice combinatorially since the $c$ factor for the root then has the same form as the $c$ factors for every other vertex.  Specifically,
\begin{equation}\label{eq anomalous dim}
\gamma(x) = \sum_{\substack{t \text{ plane }\\\text{rooted tree}}} x^{|t|} \sum_{\tau \text{ tubing of }t} \prod_{v\in V(t)} c_{b(v,\tau)-1}\,.
\end{equation}
In the next section, we use tree-tubing to give an explicit and constructive combinatorial interpretation of the  coefficients $g_n$ in the expansion of the anomalous dimension,
\begin{align} \gamma(x)=\sum_{n\geq 1} g_n\, x^n\,, \label{eq:gn} \end{align}
in the case where $F(\rho)$ is the reciprocal of a polynomial with nonnegative rational roots.

\section{Combinatorial interpretation of the anomalous dimension}

The perturbative solution of \eqref{eq tubing soln} is still not fully combinatorial because the $c_i$ coefficients are not given any combinatorial meaning. We will illustrate the combinatorial interpretation of these coefficients using the Yukawa example of \cite{bkerfc} and \cite{BD}.

\begin{example}\label{eg yukawa}
  In the Yukawa example of \cite{bkerfc} and \cite{BD} the Dyson-Schwinger equation is \eqref{DSE k=1 s=-2} with
  \[
  F(\rho) = \frac{1}{\rho(2-\rho)} = \sum_{i \geq 0} \frac{\rho^{i-1}}{2^{i+1}}\,,
  \]
  so the anomalous dimension \eqref{eq anomalous dim} in this case is
  \begin{align*} \gamma(x) & = \sum_{\substack{t \text{ plane }\\\text{rooted tree}}} x^{|t|} \sum_{\tau \text{ tubing of }t} \prod_{v\in V(t)} 2^{-b(v,\tau)} \\
 & = \sum_{\substack{t \text{ plane }\\\text{rooted tree}}} x^{|t|} \sum_{\tau \text{ tubing of }t} 2^{-\sum_{v \in V(t)}b(v,\tau)} \\
 & = \sum_{\substack{t \text{ plane }\\\text{rooted tree}}} x^{|t|}2^{1-2|t|} N(t)\,, \end{align*}
  where we used that $\sum_{v\in V(t)}b(v,\tau)$ is the number of tubes of $\tau$, namely $2|t|-1$, since every tube has a unique root and $N(t)$ is the number of tubings of $t$.

  After pulling out one factor of $2$ and scaling the coupling $x$ by $1/4$,
 this means that the anomalous dimension becomes simply the generating series for tubings.  By \cite{BCEFNOHY} Section 5 there is a size-preserving bijection between tubings of plane rooted trees and rooted connected chord diagrams.  Therefore, after scaling, the anomalous dimension in the Yukawa case is the generating series for rooted connected chord diagrams.  This was observed without proof in \cite{bk30}, and was an important insight in \cite{BD}, as this combinatorial interpretation for the anomalous dimension allowed for a resummation of all the instantons. For this resummation it was critical that the generating function of rooted connected chord diagrams fulfills a simple functional equation. A resurgence-inspired framework put forward by the first author based on generating functions and a differential operator that extracts asymptotic information from such functional relations \cite{Bgenfun} was critical to achieve this. 
This framework led to complete asymptotic information on the functional equation for the generating function of rooted connected chord diagrams and thereby to a resummable all-instanton solution.

\end{example}

Until now, a combinatorial explanation at this level was not known in the other cases.  The goal of this section is to give one and derive a system of differential equations from the interpretation.

Suppose now that the Mellin transform $F(\rho)$ is the reciprocal of a polynomial and has all its roots nonnegative and rational.  After factoring out the $1/\rho$ for the simple pole, all roots of the remaining polynomial will then be positive and rational.  Scaling appropriately, then, we can write such a Mellin transform in the following form 
\begin{equation}\label{eq nice mellin}
  F(\rho) = \frac{1}{\rho(1-a_1\rho)(1-a_2\rho)\cdots (1-a_m\rho)}\,,
\end{equation}
where $a_1, \ldots, a_m \in \mathbb{Z}$ and $0<a_1\leq a_2\leq \cdots \leq a_m$.  For the rest of the paper we will assume our Mellin transform is of this form. To return to a more physical scaling and thus agree with past work, see Remark~\ref{rem:rescaling}.  We know from elementary enumeration that the generating series for $k$-ary strings (that is, strings or words from an alphabet of $k$ letters) is $1/(1-kx)$.  We also know from elementary enumeration that if $A(x)$ and $B(x)$ are generating series for $\mathcal{A}$ and $\mathcal{B}$ respectively then $A(x)B(x)$ is the generating series for $\mathcal{A}\times \mathcal{B}$ where the size of an ordered pairs is the sum of the sizes of its two components~\cite[Part.~A]{FS}.  %
Additionally, suppose we have an ordered pair of a word on the alphabet $\Omega$ and a word on the alphabet $\Omega'$ where $\Omega\cap \Omega'=\emptyset$.  By concatenating the two words, we can view any such ordered pair as a single word on the alphabet $\Omega\cup \Omega'$ where all letters from $\Omega$ must appear before all letters from $\Omega'$.
Taking all this together we get that \eqref{eq nice mellin}, after multiplying by $\rho$ to remove the pole, 
is the generating series (in $\rho$) for words with letters from $\Omega_1\cup \Omega_2\cup\cdots \Omega_m$ where the $\Omega_i$ are disjoint alphabets, $\Omega_i$ has $a_i$ letters, and where the words have all $\Omega_i$ letters before any $\Omega_j$ letters for all $i<j$.  Given such $\Omega_i$, call a word on $\Omega_1\cup \cdots \cup \Omega_m$ where all letters from $\Omega_i$ occur before any letters from $\Omega_j$ for each $i<j$ an \emph{admissible} word.

Let $w_n$ be the number of admissible words of length $n$.
Using the interpretation of  $1/((1-a_1\rho)(1-a_2\rho)\cdots (1-a_m\rho))$ as the generating series for admissible words, calculate
\begin{align*} \gamma(x) & = \sum_{\substack{t \text{ plane }\\\text{rooted tree}}} x^{|t|} \sum_{\tau \text{ tubing of }t} \prod_{v\in V(t)} c_{b(v,\tau)-1} \\
 & = \sum_{\substack{t \text{ plane }\\\text{rooted tree}}} x^{|t|} \sum_{\tau \text{ tubing of }t} \prod_{v\in V(t)} w_{b(v,\tau)-1} \end{align*}

Note that we have a formula for $w_n$, but this will not be crucial in what follows.  Namely, by the required order on the letters in an admissible word, we can sum over all possibilities for which $\Omega_i$ each letter is in by summing over $1\leq i_1\leq i_2\leq \cdots \leq i_n \leq m$, and for each position then we simply need to take the size of the alphabet, giving 
\begin{equation}\label{eq number of admissible words}
w_n = \sum_{ 1\leq i_1\leq i_2\leq \cdots \leq i_n \leq m}a_{i_1}a_{i_2}\cdots a_{i_n}.
\end{equation}

Now we want to give a combinatorial understanding of $\prod_{v\in V(t)} w_{b(v,\tau)-1}$.

Let $\tau$ be a tubing of a plane rooted tree $t$.

Note that we can divide the tubes of $\tau$ into two kinds \emph{upper tubes} and \emph{lower tubes} in the following way: every tube other than the outermost tube comes about as half of a bipartition of some subtree.  In that bipartition one tube contains the root of the subtree, call this an upper tube, and one does not, call this a lower tube.  By convention take the outermost tube to be a lower tube.  

Furthermore, for any vertex $v$, the set of tubes for which $v$ is the root have a strictly nested structure giving them a linear order and the outermost tube for which $v$ is the root is a lower tube (either the lower tube for the bipartition corresponding to the edge from $v$ to its parent, or the outermost tube of the whole tubing if $v$ is the root of the whole tree), while all other tubes for which $v$ is the root are upper tubes.  Therefore $b(v, \tau)-1$ is the number of upper tubes of $\tau$ for which $v$ is the root.

Now label each upper tube with a letter from $\Omega_1\cup\cdots\cup \Omega_m$ with the restriction that whenever $i<j$ and we have a letter from $\Omega_i$ and a letter from $\Omega_j$ which label two tubes with the same root, then the letter from $\Omega_i$ is on a tube which is inside the tube with the letter from $\Omega_j$.  With these constraints, the letters on tubes with the same root give an admissible word when the letters are formed into a word by the innermost to outermost order of the tubes, and every way of assigning admissible words of length $b(v,\tau)$ to each vertex is obtained in this way.\footnote{We could just as well note that every vertex is in one tube of size $1$ in $\tau$ and so also $b(v,\tau)-1$ is the number of tubes of size $>1$ of $\tau$ for which $v$ is the root.  Thus we could replace ``upper tube'' with ``tube of size $>1$'' in the labeling by letters and obtain all the same results.  However, as shown in \cite{OHthesis}, the upper tubes are structurally nicer, and so, they are a preferable choice.}

Encapsulate such a labeling of the tubes of a tubing in a definition.
\begin{definition}
  Let $\Omega_1, \ldots \Omega_m$ be disjoint alphabets with $|\Omega_i|=a_i$.  
  Let $\tau$ be a tubing of a plane rooted tree $t$.  Call a labeling of the upper tubes of $\tau$ with letters from $\Omega_1\cup \cdots \cup \Omega_m$a \emph{compatible labeling} (or a $a_1, \ldots, a_m$-compatible labeling, or a $\Omega_1, \ldots, \Omega_m$-compatible labeling, when distinguishing the sizes of alphabets is necessary) if whenever $i<j$ and there is a letter from $\Omega_i$ and a letter from $\Omega_j$ which label two tubes with the same root, then the letter from $\Omega_i$ is on a tube which is inside the tube with the letter from $\Omega_j$.
\end{definition}
Let $L(\tau)$ be the number of compatible labelings of $\tau$.
Then continuing the calculation of $\gamma(x)$
\begin{align*} \gamma(x) & = \sum_{\substack{t \text{ plane }\\\text{rooted tree}}} x^{|t|} \sum_{\tau \text{ tubing of }t} \prod_{v\in V(t)} w_{b(v,\tau)-1} \\
 & = \sum_{\substack{t \text{ plane }\\\text{rooted tree}}} x^{|t|} \sum_{\tau \text{ tubing of }t} L(\tau). \end{align*}
In particular, $\gamma(x)$ is the generating series for compatible labelings of tubings of plane rooted trees.

Encapsulating the result, we just proved the following proposition:
\begin{prop}\label{prop nice mellin result}
  Let $\Omega_1, \ldots \Omega_m$ be disjoint alphabets with $|\Omega_i|=a_i$.  The anomalous dimension of the perturbative solution to \eqref{DSE k=1 s=-2} with Mellin transform \eqref{eq nice mellin} is the generating series for compatible labelings of tubings of plane rooted trees.
\end{prop}

\begin{example}\label{eg phi3}
  Consider the $\phi^3$ example from \cite{bkerfc} and \cite{BDM}.  The Dyson-Schwinger equation is \eqref{DSE k=1 s=-2} with Mellin transform
  \[
  \frac{1}{\rho(3-\rho)(2-\rho)(1-\rho)}.
  \]
  This is a Mellin transform which can be put in the form of \eqref{eq nice mellin} after a suitable scaling of $\rho$.  Specifically,
  \begin{gather*} \frac{1}{\rho(3-\rho)(2-\rho)(1-\rho)} = \frac{1}{6\rho(1-\frac{1}{3}\rho)(1-\frac{1}{2}\rho)(1-\rho)} \\
 = \frac{1}{36 \rho'(1-2 \rho')(1-3\rho')(1-6\rho')} \end{gather*}
  which, after factoring out the $1/36$ is
  in the form of \eqref{eq nice mellin} with $a_1 = 2$, $a_2=3$, $a_3 = 6$, and with $\rho$ replaced by $\rho'=\rho/6$.  

  Proposition~\ref{prop nice mellin result} then tells us that, after these scalings,\footnote{See Remark~\ref{rem:rescaling} on reversing the scalings in order to compare with \cite{bkerfc} and \cite{BDM} directly.} the anomalous dimension for the perturbative solution to the Dyson-Schwinger equation is the generating series of compatible labelings of tubings of plane rooted trees with three alphabets, one with 2 letters, one with 3 letters and one with 6 letters with the compatibility restriction, so tubes with letters from smaller alphabets appear inside tubes with the same root and with letters from larger alphabets.

  In particular, $2,3,6$-compatible labelings of tubings of plane rooted trees give a combinatorial interpretation for the sequence given at the end of section 2 of \cite{BDM} which the authors observed had no known interpretation at the time.

  To make the tubing construction more concrete let us work out the first few terms explicitly.  The tree with a single vertex has a single tubing consisting of just one tube.  This tube is by convention a lower tube and so, it does not get a letter.  Therefore the coefficient of $x^1$ in $\gamma(x)$ is $1$ in this model.  There is one tree with two vertices and one tubing of it, illustrated on the left in Figure~\ref{fig example worked out}.  The only upper tube in this case is the tube containing only the root, so there is one upper tube rooted at the root and no upper tubes rooted at the other vertex. These numbers are illustrated in the figure.  Thus there is only one tube that gets a letter, it can have any of the $2+3+6=11$ letters, as we expect from \eqref{eq number of admissible words}, so the coefficient of $x^2$ in $\gamma(x)$ is $11$.  

  \begin{figure}
      \centering
      \includegraphics[scale=1.5]{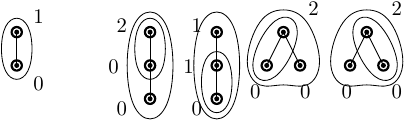}
      \caption{Tubings of rooted trees on two and three vertices with the number of upper tubes rooted at each vertex indicated.}
      \label{fig example worked out}
  \end{figure}

  Next consider the coefficient of $x^3$ in $\gamma(x)$.  There are two rooted trees on three vertices, each of which has two tubings as illustrated on the right in Figure~\ref{fig example worked out}.  The number of upper tubes rooted at each vertex is indicated in the figure.   In each case can count the number of ways to put letters at the upper tubes around each vertex via the formula of \eqref{eq number of admissible words}, but to keep this example explicit let us count directly. For the third tubing from the right, there are $11$ ways to label each upper tube since they are rooted at different vertices so there is no constraint on which letters can be chosen.  For each of the other three tubings we must label a pair of tubes with the same root.  If the outer of the two tubes is labelled by a letter from $\Omega_3$ then there is no constraint on the label of the inner of the two tubes.  If the outer of the two tubes is labelled by a letter from $\Omega_2$ then the label of the inner of the two tubes must be from $\Omega_2\cup \Omega_3$ while if the label of the outer of the two tubes is from $\Omega_1$ then the label of the outer of the two tubes must also be from $\Omega_1$.  Consequently there are $6\cdot 11 + 3 \cdot (2+3) + 2\cdot 2 = 85$ compatible labelings of such a tubing.

  Altogether, this gives that the coefficient of $x^3$ in $\gamma(x)$ is $11^2 + 3\cdot 85 = 376$, agreeing with the explicit calculations of \cite{bkerfc, BDM}.
\end{example}

As well as providing previously unknown combinatorial interpretations, Proposition~\ref{prop nice mellin result} also gives a system of differential equations for $\gamma(x)$ that can be used to compute its coefficients efficiently and which is well suited for the resurgent approach.

To build this system of differential equations, first consider how to recursively build tubings.  As discussed in \cite{BCEFNOHY}, by the definition of tubing, the outermost tube of any tubing is partitioned into subtubes and each of those subtubes along with all the tubes within determines a tubing of a connected subgraph of the original tree.  Furthermore, take any two tubings $\tau_1$ and $\tau_2$ of plane rooted trees $t_1$ and $t_2$ respectively and then choose an insertion place in $t_2$.  Adding an edge at the insertion point connecting to the root of $t_1$ builds a new tree $t$ and taking all the tubes of $\tau_1$ and $\tau_2$ along with the outer tube for $t$ gives a tubing $\tau$ of $t$, and every tubing of every rooted tree can be formed uniquely in this way.  Furthermore, if a plane rooted tree has $n$ vertices then it had $2n-1$ insertion places.  Therefore, if $N(n)$ is the number of tubings of rooted trees on $n$ vertices then
\[
N(n) = \sum_{i=1}^{n-1} N(n-i) (2i-1)N(i) \qquad N(1)=1
\]
or at the level of generating series, if $T(x)=\sum_{n\geq 1}N(n)x^n$ then
\begin{equation}\label{eq tube de}
T(x) = x + T(x) \left(2x\frac{d}{dx} -1\right)T(x)\,.
\end{equation}
The differential operator $2x\frac{d}{dx}-1$ serves to pull out the coefficient $2i-1$, and hence serves to count the insertion places.  This is a special case of the general interpretation of $x\frac{d}{dx}$ as a \emph{pointing operator} in enumerative combinatorics \cite[Ch.~A.I]{FS}.  The recurrence \eqref{eq tube de} for tubings of plane rooted trees was used in \cite{BCEFNOHY} to give a bijection between tubings of plane rooted trees and rooted connected chord diagrams, which are classically known to have the same recurrence \cite{NWchord}.

We will use this same construction but keep track of compatible labelings to prove the following theorem.
\begin{thm}\label{thm comb de}
  Let $\Omega_1, \ldots \Omega_m$ be disjoint alphabets with $|\Omega_i|=a_i$. Let $L_j(x)$ be the generating series for $\Omega_1, \ldots \Omega_m$-compatible labelings of tubings of plane rooted trees for which the label on the outermost tube comes from $\Omega_j$.   Then the anomalous dimension $\gamma(x)$ of the perturbative solution to \eqref{DSE k=1 s=-2} with Mellin transform \eqref{eq nice mellin} satisfies the following system of differential equations.
  \begin{align*} \gamma(x) & = x + \sum_{k=1}^m L_k(x)\,, \\
 L_j(x) & = a_j \gamma(x) \left(2x\frac{d}{dx}-1\right)\left(x + \sum_{i=1}^{j}L_i(x)\right) \qquad \text{for $1\leq j\leq m$}\,. \end{align*}
\end{thm}

\begin{proof}
  As before let $\tau_1$ and $\tau_2$ be compatibly labelled tubings of plane rooted trees $t_1$ and $t_2$ respectively and choose an insertion place in $t_2$.  Construct $\tau$ and $t$ as above by inserting $t_1$ into $t_2$ at the chosen insertion place.  The labelings of the tubes of $\tau_1$ and $\tau_2$ remain compatible in $\tau$.  The only thing missing is a label for the upper tube of the partition of the outermost tube, which had previously been a lower tube and is now an upper tube.  Call this label the outermost label of $\tau$.  The outermost label of $\tau$ is constrained to be from an alphabet $\Omega_j$ with $j\geq i$ where the outermost label of $\tau_2$ is from $\Omega_i$, or if $t_2$ is a single vertex and so has no labelled tubes then the outermost label of $\tau$ is unconstrained.  Furthermore, this is the only constraint on the outermost label of $\tau$ and every compatible labeling can be obtained uniquely in this way. The labels from $\tau_1$ put no constraint on the outermost label of $\tau$ since they do not share roots.

  As above we can encode this decomposition by generating series.  Every compatible labeling of a tubing of a plane rooted tree has its outermost label from exactly one of the alphabets except for the tree with one vertex and so
  \[
  \gamma(x) = x + \sum_{k=1}^m L_k(x)\,.
  \]
  For a compatible labeling with the outermost label from $\Omega_j$, the decomposition described above tells us we have $a_j$ choices for the outermost label, we can take any compatible labeling of a tubing of a plane rooted tree for $\tau_1$, hence giving a factor of $\gamma(x)$ in the generating series, and we can take either a compatible labeling of a tree with one vertex or any compatible labeling with outermost label from $\Omega_1\cup\cdots\cup\Omega_j$ for $\tau_2$, where as before the differential operator $2x\frac{d}{dx}-1$ incorporates the count of the number of insertion places, and so
\[
L_j(x) = a_j \gamma(x) \left(2x\frac{d}{dx}-1\right)\left(x+\sum_{i=1}^{j}L_i(x)\right),
\]
as desired.
\end{proof}

\begin{remark}
Very similar systems of equations were already known in the work of Marc Bellon, see equations 26 and 27 of \cite{Befficient}. There are minor differences, with the different equations of the system in Bellon's case corresponding to different variables in the Mellin transform, one pole per variable, while we get the different equations in our systems from different poles in the same variable, leading to our triangular form. However, the two contexts agree on their common intersection, and so, despite being developed independently, Bellon's work is an important predecessor of ours, and one which already appreciated the suitability of such systems for investigating asymptotics.

The main novelty of our approach is not so much the equations of Theorem~\ref{thm comb de} themselves, but the fact that the combinatorial interpretation means that the equations fall out automatically from the combinatorial structure -- they are essentially a trivial rewriting of the recursive structure of tubings -- and that each piece of the equation has a combinatorial meaning, providing insight into both specializations and generalizations. 
\end{remark}

\begin{cor}
\label{cor:nonlinear}
The set of differential equations in Theorem~\ref{thm comb de} can be combined into a single nonlinear ODE
\begin{eqnarray} \prod_{j=1}^m \left[1-a_j \gamma(x)\left(2x\frac{d}{dx}-1\right)\right]\gamma(x)=x\,. \label{eq:ode} \end{eqnarray}
Up to scalings, this is the form in which they were studied previously in \cite{bkerfc,bk30,BD,BDM,BB}. The comparison with other scalings in the literature for the cases in Examples~\ref{eg de}--\ref{eg in} is discussed below in Remark~\ref{rem:rescaling}.
\end{cor}
\begin{proof}
    Let $\gamma_j:=x+\sum_{k=1}^j L_k$. For $j=1$, we have $x=\gamma_1-L_1$, and since $$L_1=a_1\gamma \left(2x\frac{d}{dx}-1\right)\gamma_1,$$ we have $x=\left[1-a_1\gamma \left(2x\frac{d}{dx}-1\right)\right]\gamma_1$. Similarly, $\gamma_{j+1}-L_{j+1}=\gamma_{j}$ implies $$\left[1-a_{j+1}\gamma \left(2x\frac{d}{dx}-1\right)\right]\gamma_{j+1}=\gamma_j.$$ The result follows by induction, noting that $\gamma_m=\gamma$,  by definition.
\end{proof}

\begin{remark}
\label{remark gamma program}
An advantage of the differential equation system form in 
Theorem~\ref{thm comb de} is that it provides a convenient way to compute the coefficients of the power series $\gamma(x)=\sum_{n\geq 1} g_n x^n$.
To do so, we substitute power series ansätze for $\gamma(x)$ and $L_j(x)$ into the differential equation system. This yields a system of nonlinear difference equations for the coefficients that can be solved recursively
(see Section~\ref{subsubchain} for a simple example of this procedure). We implemented this recursion as the python program \texttt{gamma.py} which is included in the ancillary files to the arXiv version of this article.
In a short time, this program computes the numbers $g_n$ for various input parameters $a_1,\ldots,a_m \in \Z$ up to $n\approx 1000$.
\end{remark}

As mentioned in the introduction, special cases of the differential equations system given by Theorem~\ref{thm comb de} or, equivalently, the nonlinear ODE in Corollary~\ref{cor:nonlinear}
were studied in \cite{bkerfc, BD, BDM, BB}. We call them the \emph{Yukawa example}, the \emph{$\phi^3$ example}, and \emph{the intermediate example}.
The precise matching of the present prefactor and scaling conventions with the ones of previous articles is discussed below in Remark~\ref{rem:rescaling}.

\begin{example}\label{eg de}
  For the Yukawa example, continuing Example~\ref{eg yukawa}, after rescaling to the form of~\eqref{eq nice mellin} we get $a_1=1$ and $m=1$.  The system of differential equations then clearly simplifies into the single differential equation studied in \cite{bkerfc,bk30,BD}:
\[
 \left[1-\gamma_{\rm Yukawa}(x)\left(2x\frac{d}{dx}-1\right)\right]\gamma_{\rm Yukawa}(x)=x\,.
\]
Its unique power series solution (see 
\cite[\href{https://oeis.org/A000699}{A000699}]{oeis})
starts with 
\begin{eqnarray} \gamma_{\rm Yukawa}(x) = \sum_{n \geq 1} g^{\rm Yukawa}_n x^n = x +x^2 +4 x^3+27 x^4+\dots   \label{eq:hopfm1} \end{eqnarray}
This particular example is the only case, where a complete description of the asymptotic behavior of the coefficients $g^{\rm Yukawa}_n$ is available \cite{Bgenfun,BD}. 
The numbers $g^{\rm Yukawa}_n$ count {\it connected chord diagrams}, which follows from writing the ODE 
that $\gamma_{\rm Yukawa}(x)$ fulfills as a difference equation for the coefficients $g^{\rm Yukawa}_n$ (see \cite{NWchord} for the combinatorial analysis of this difference equation).
In Section~\ref{sec:asymptotics}, we discuss the asymptotic behavior of the coefficients of general power series solutions to the differential equation system in Theorem~\ref{thm comb de}. 

To describe the asymptotic behavior of sequences such as $g^{\rm Yukawa}_n$, we need some notation. For a number sequence $f_n$ and a family of sequences $c_n^{(0)},c_n^{(1)},\ldots$, the \emph{Poincar\'e asymptotic expansion notation}
`$ f_n \sim \sum_{k \geq 0} c_n^{(k)} \text{ as } n\rightarrow \infty $'
means that 
$$
\limsup_{n \rightarrow \infty} \left| \frac{f_n - \sum_{k=0}^{R-1} c_n^{(k)}}{c_n^{(R)}}\right| < \infty \text{ for all } R \geq 0.
$$
The asymptotic expansion of the coefficients $g_n^{\rm Yukawa}$ starts as
\[
    g_n^{\rm Yukawa}\sim \frac{1}{e} \frac{2^{n+1/2} \Gamma\left(n+\frac{1}{2}\right)}{\sqrt{2\pi}}\left(1-\frac{5/4}{n-1/2}-\dots\right)\qquad  \text{ as } n \rightarrow \infty
,
\]
where the classical $\Gamma$-function is defined by the integral $\Gamma(z) = \int_0^\infty y^{z-1} e^{-y} \dd y$. 

This is one of the rare cases where the prefactor $\frac{1}{e}$ is known exactly thanks to 
a functional equation that the associated generating function fulfills (see \cite[\S~7.1]{Bgenfun}).

\end{example}

\begin{example}
  For the $\phi^3$ example, continuing Example~\ref{eg phi3}, after rescaling to the form of \eqref{eq nice mellin} we have $a_1=2$, $a_2=3$, $a_3=6$, $m=3$. The system of differential equations in Theorem~\ref{thm comb de} reduces to
  \begin{align*} \gamma(x) & = x + L_1(x) + L_2(x) + L_3(x) \\
 L_1(x) & = 2 \gamma(x) \left(2x\frac{d}{dx}-1\right)\left(x + L_1(x)\right) \\
 L_2(x) & = 3 \gamma(x) \left(2x\frac{d}{dx}-1\right)\left(x + L_1(x) + L_2(x)\right) \\
 L_3(x) & = 6 \gamma(x) \left(2x\frac{d}{dx}-1\right)\left(x + L_1(x) + L_2(x) + L_3(x)\right) \end{align*}
  This system is equivalent to 
  a third-order nonlinear ODE with quartic nonlinearity as in Corollary~\ref{cor:nonlinear}:
\[
\left[1-2 \gamma(x)\left(2x\frac{d}{dx}-1\right)\right]\left[1-3 \gamma(x)\left(2x\frac{d}{dx}-1\right)\right]\left[1-6\gamma(x)\left(2x\frac{d}{dx}-1\right)\right]\gamma(x)=x\,.
\]
The first coefficients of the unique power series solution 
(see \cite[\href{https://oeis.org/A051862}{A051862}]{oeis})
are
\begin{eqnarray} \gamma(x)= x+11x^2 +376x^3+20241 x^4+ \dots   \label{eq:hopf236} \end{eqnarray}

\end{example}
\begin{example} \label{eg in}
  For the intermediate example of \cite{BB} the Mellin transform is $1/(\rho(2-\rho)(1-\rho))$. So, after rescaling to the form of \eqref{eq nice mellin} we have $a_1 = 1$, $a_2=2$, $m=2$. In this case the system is
    \begin{align*} \gamma(x) & = x + L_1(x) + L_2(x) \\
 L_1(x) & = \gamma(x) \left(2x\frac{d}{dx}-1\right)\left(x + L_1(x)\right) \\
 L_2(x) & = 2 \gamma(x) \left(2x\frac{d}{dx}-1\right)\left(x + L_1(x) + L_2(x)\right)\,. \end{align*}
    It is equivalent to a second-order ODE with cubic nonlinearity,
\[
\left[1- \gamma(x)\left(2x\frac{d}{dx}-1\right)\right]\left[1-2 \gamma(x)\left(2x\frac{d}{dx}-1\right)\right]\gamma(x)=x\,,
\]
whose unique power series solution starts with
\begin{eqnarray} \gamma(x)= x+3x^2 +30x^3+483 x^4+ \dots   \label{eq:hopf12} \end{eqnarray}
\end{example}

\begin{remark} Examples \ref{eg de}--\ref{eg in} agree with previous works \cite{bkerfc,bk30,BD,BDM,BB} after suitable rescalings of both $\gamma$ and $x$.
For example, Eq. (11) in \cite{BD} has $\tilde{\gamma}(\alpha)=2\gamma(-\alpha/4)$ with $\gamma(x)$ as in Eq. \eqref{eq:hopfm1} here. To compare with \cite{bkerfc,bk30,BDM} we have 
$\tilde{\gamma}(\alpha)=6\gamma(-\alpha/36)$ with $\gamma(x)$ as in Eq. \eqref{eq:hopf236} here. To compare with Eq. (29) in \cite{BB} we have 
$x\,\tilde{g}_0(x)=2\gamma(x/4)$ with $\gamma(x)$ as in Eq. \eqref{eq:hopf12} here.
    \label{rem:rescaling}
\end{remark}

\begin{remark}[Holonomicity of $\gamma$ in Theorem~\ref{thm comb de}]
Holonomic or $D$-finite functions and, equivalently, $P$-recursive sequences have an important role in algebraic combinatorics. 
The case $m=1$ (i.e.~the case discussed in Example~\ref{eg de}) of Theorem~\ref{thm comb de} provides one of the rare instances when a sequence is proven not to be $P$-recursive \cite{MR1979786}. 
It is reasonable to suspect that $\gamma$ is non-holonomic for most sets of parameters $a_1,\ldots,a_m\in \Z$.
\end{remark}

\subsection{Rainbows, chains, and other approximations}
\label{sec:approxs}
Certain further restrictions on the class of Feynman diagrams used in the approximation to the perturbative expansion that have been previously considered in the physics literature and specifically in previous analysis of \eqref{DSE k=1 s=-2} in \cite{BD} are the \emph{rainbow} and \emph{chain} approximations, the latter of which is important for understanding what are known in physics as \emph{renormalons}.  However, while these two approximations are simpler and so not as hard to understand in and of themselves, it was not entirely clear how to understand the differences between the behavior of these approximations since they did not live within a common framework.

One benefit of the system of differential equations from Theorem~\ref{thm comb de} is that the rainbow and chain approximations (and other potential approximations of less clear physical significance) inherit from the combinatorial interpretation by compatibly labelled tubings their own systems of differential equations obtained by suitably modifying the system of Theorem~\ref{thm comb de}.  This makes it easy to read off the key differences and similarities between these approximations and see how these differences affect the growth rate of the coefficients. 

The reason that this happens is that both the rainbow and the chain approximation are obtained by restricting the Feynman diagrams that contribute in a way that corresponds to restricting the shape of either the inserted tree or the insertion place into which it is inserted when forming the tubing (as described in the proof of Theorem~\ref{thm comb de}).  Any approximation that corresponds to restricting the shape of the inserted tree, the tree being inserted, or the location of the insertion place inherits a suitably restricted systems of differential equations.  Here it is important to notice that the rooted tree being tubed is not just any rooted tree but is the insertion tree showing the tree structure of the insertions of the 1-loop graph into itself (see \cite{BCEFNOHY}).

\subsubsection{Rainbow approximation}
For example, for rainbow approximations, the insertion trees are ladders, that is trees where each vertex has at most one child.  This is a restriction on the possible insertion places: there is exactly one insertion place, at the bottom of the ladder.  The ladder structure itself is then built recursively with no further restrictions needed.  Consequently, the only change required in the system of differential equations is that the differential operator is removed, so we have for the rainbow approximation associated to the Mellin transform in \eqref{eq nice mellin}:
\begin{align*} \gamma_{\rm rainbow}(x) & = x + \sum_{k=1}^m L_k^{\rm rainbow}(x)\,, \\
 L_j^{\rm rainbow}(x) & = a_j \gamma_{\rm rainbow}(x) \left(x + \sum_{i=1}^{j}L_i^{\rm rainbow}(x)\right) \qquad \text{for $1\leq j\leq m$}\,. \end{align*}

These are purely algebraic equations.  Let us consider the particular case\footnote{For the approximations discussed in Section \ref{sec:approxs}, we are particularly interested in the different overall growth behavior of the coefficients of the expansion and the $m=1$ case is sufficient to illustrate this difference with tidy, relatively easy to understand equations.} 
when $m=1$.  In this case the algebraic (albeit nonlinear) equation is %
$\gamma(x)(1- a_1 \gamma(x))=x$, with unique solution (scaling to $a_1=1$ for convenience)
\[
\gamma_{\rm rainbow}(x)=\frac{1-\sqrt{1-4 x}}{2}=\sum_{n\geq 1} C_{n-1} x^{n}\,.
\]
This expansion has a finite radius of convergence, $\frac{1}{4}$, and the coefficients are expressed in terms of the Catalan numbers $C_n$.

\subsubsection{Chain approximation}
\label{subsubchain}
For the chain approximation (as inserted into one outer graph), the insertion trees are corollas, that is trees where every child of the root is a leaf.  This is a restriction on the form of the tree being inserted -- only single vertex trees can be inserted -- so the factor $\gamma(x)$ is replaced by $x$.  Additionally, the insertion places must be at the root, so for a corolla with $n$ vertices, we have $n$ insertion places (before all $n-1$ children of the root, between any two adjacent children, and after all children).  This results in changing the differential operator from $\left(2x\frac{d}{dx}-1\right)$ to $x\frac{d}{dx}$.  Altogether, the system of differential equations for the chain approximation associated to the Mellin transform in~\eqref{eq nice mellin} is
  \begin{align*} \gamma_{\rm chain}(x) & = x + \sum_{k=1}^m L_k^{\rm chain}(x)\,, \\
 L_j^{\rm chain}(x) & = a_j x \left(x\frac{d}{dx}\right)\left(x + \sum_{i=1}^{j}L_i^{\rm chain}(x)\right) \qquad \text{for $1\leq j\leq m$}\,. \end{align*}
This is a system of {\it linear} ODEs. For $m=1$ the equation is $\gamma_{\rm chain}(x)-a_1 x^2\, \gamma_{\rm chain}^\prime(x)=x$.
We can find the power series solution to this ODE 
by substituting the ansatz $\gamma_{\rm chain}(x) = \sum_{n\geq 1} g^{\rm chain}_n x^n$ into the ODE.
This immediately results in the difference equation
$ g^{\rm chain}_n= a_1 n g^{\rm chain}_{n-1} $
for all $n\geq 2$, with the initial 
condition $g^{\rm chain}_1 = 1$.
If we fix $a_1=1$, we get the power series solution: %
\[
    \gamma_{\rm chain}(x)= 
     \sum_{n\geq 1} g^{\rm chain}_n x^{n}=
    \sum_{n\geq 1}(n-1)! \, x^{n}\,.
\]
The right-hand-side is the asymptotic expansion of $-e^{-1/x}\Gamma(0,-1/x)$ as $x\to 0^+$, where $\Gamma(a, b)$ is the incomplete gamma function. Notice that unlike the rainbow approximation the expansion coefficients $g_n^{\rm chain}$ are factorially divergent. 
Stirling's formula for the asymptotic growth rate of $n!$ shows that the exponentially growing prefactor in the asymptotic expansion of $g_n^{\rm chain}$ is half the one of the original $a_1=1$ Yukawa case in \eqref{eq:hopfm1}.
From a resurgence perspective, this means that the  associated Borel singularity of $\gamma_{\rm chain}$ is twice as far from the origin as for \eqref{eq:hopfm1}
(see \cite{BD}).

Comparing the systems for the rainbow and chain approximations we can see immediately why they behave so differently from the more general approximation of Theorem~\ref{thm comb de}.  For the rainbow the differential operator is gone, while for the chain there is a small modification of the differential operator compared to Theorem~\ref{thm comb de}, but the key difference is that one of the component series is truncated to just $x$.  

The fact that the chain approximation has a differential equation where the appearance of the anomalous dimension on the right hand sides is replaced by the truncation to $x$ was observed in \cite{Bdsems} (see equation 7.2 and below) in that case phrased at the level of the Green function and hence giving an equation like the renormalization group equation but with this truncation of the anomalous dimension. This also underlies the special role of chains (there called corollas) in \cite{DFY}.

\subsubsection{Other intermediate approximations}\label{sec other intermediate approx}

There are other approximations with natural combinatorial interpretation.  For instance, if we change the differential operator from $2x\frac{d}{dx}-1$ to $x\frac{d}{dx}$ but keep the $\gamma(x)$ factor, then we have an approximation that is intermediate between the full system and the chain approximation:
  \begin{align*} \gamma_{\rm odd}(x) & = x + \sum_{k=1}^m L_k^{\rm odd}(x)\,, \\
 L_j^{\rm odd}(x) & = a_j \gamma_{\rm odd}(x) \left(x\frac{d}{dx}\right)\left(x + \sum_{i=1}^{j}L_i^{\rm odd}(x)\right) \qquad \text{for $1\leq j\leq m$}\,. \end{align*}
In this case, we can interpret the equation as telling us we insert only into odd numbered insertion places and no other restrictions compared to the full system.  A quick induction shows that this gives all plane rooted trees, but the insertion places are exactly those insertion places under vertices at even level in the tree (where the root is at level 0) and so when a tube is bipartitioned, the bipartition is always by cutting an edge which is under a vertex at even level, relative to the subtree induced by the tube under consideration.  Rephrasing this in terms of Feynman diagrams, we build the same diagrams as in the full system, but each diagram's contribution is somewhat simplified as only some tubings contribute.

This is a system of nonlinear ODEs and for $m=1$ the equation is the nonlinear ODE $\gamma_{\rm odd}(x)(1-a_1 x\, \gamma_{\rm odd}^\prime(x))=x$, with a unique factorially divergent power series solution (again scaling to $a_1=1$):
\[
    \gamma_{\rm odd}(x)= \sum_{n\geq1} g^{\rm odd}_n x^n = x+x^2+3x^3+14x^4+85x^5+621x^6+\dots
\]
The coefficients $g^{\rm odd}_n$ are listed in the OEIS as \cite[\href{https://oeis.org/A088716}{A088716}]{oeis}
and they grow factorially. It is interesting to note that this series appears in the study of exponential Dyson-Schwinger equations~\cite{PHB}.

The asymptotic expansion of the numbers $g^{\rm odd}_n$ starts as
\[
    g^{\rm odd}_n\sim S\, (n+1)!\left(1-\frac{3}{n+1}-\frac{\frac{5}{2}}{n(n+1)}-\dots\right) \text{ as } n \rightarrow \infty\,,
\]
with the prefactor $S=0.217950789447151...$ which is only known numerically (see, e.g.,~\cite[Appendix~A]{BDM} for a brief explanation on how to obtain such constants empirically). 
These coefficients have a similar, but faster, factorial rate of growth compared to the chain approximation.
As with the chain approximation, the associated Borel singularity is twice as far from the origin as the leading one of the original approximation for the Yukawa case, indicating slower growth by a factor of $2^{-n}$. 

Another intermediate approximation is obtained by taking the chain approximation and restoring the differential operator to $2x\frac{d}{dx}-1$:
  \begin{align} \gamma_{\rm chord}(x) & = x + \sum_{k=1}^m L_k^{\rm chord}(x)\,, \\
 L_j^{\rm chord}(x) & = a_j x \left(2x\frac{d}{dx}-1\right)\left(x + \sum_{i=1}^{j}L_i^{\rm chord}(x)\right) \qquad \text{for $1\leq j\leq m$}\,. \label{eq:approx2} \end{align}
In this case we can interpret the equation as telling us that we insert only single vertices (counted by the leading factor of $x$) but we can insert them into any location in the tree.  Again we get all plane rooted trees, but with a restriction on which tubings contribute, namely the only tubings that contribute are those where each lower tube (excepting the outermost tube) is a single vertex; these are called \emph{leaf tubings} in \cite{BCEFNOHY}.  In terms of Feynman diagrams, as with the other intermediate approximation we build the same diagrams as in the full system, but again each diagram's contribution is somewhat simplified as only some tubings contribute, the specific set of tubings which contribute being different between the two intermediate examples.

This approximation produces a system of {\it linear} ODEs. For $m=1$, there is just one linear ODE: $(1+a_1\, x)\gamma_{\rm chord}(x)-2a_1 x^2 \gamma_{\rm chord}^\prime(x)=x$, with a unique factorially divergent power solution (scaling to $a_1=1$):
\begin{eqnarray} \gamma_{\rm chord}(x)  = \sum_{n=1}^\infty \frac{2^{n-\frac{1}{2}}\Gamma\left(n-\frac{1}{2}\right)}{\sqrt{2\pi}}\, x^{n}\,. \label{eq:approx2b} \end{eqnarray}
This is the generating function for all chord diagrams, and the formal series is the asymptotic expansion of $\sqrt{\frac{-x}{2}}e^{-\frac{1}{2x}}\Gamma\left(\frac{1}{2}, -\frac{1}{2x}\right)$.  This contrasts with the Yukawa example (Example~\ref{eg de}) where instead we have the generating function for connected chord diagrams of size one larger.  These two examples have almost the same growth rate with the Yukawa coefficients growing only linearly faster than those of $\gamma_{\rm chord}$.

The second intermediate approximation has an additional nice property.  Since every lower tube is a single vertex in this approximation, all the upper tubes are rooted at the root of the tree.  Therefore, the number of compatible labelings of such a tubing of a tree $t$ is simply the number of admissible words of length $|t|-1$. In particular the number of compatible labelings does not depend on the shape of the tree or on the specific tubing (among those tubings under consideration at all in this approximation).  Consequently, the coefficient of $x^{n}$ in the anomalous dimension factors into the product of the number of admissible words of length $n-1$ and the number of leaf tubings of rooted trees of size $n$.  The first of these factors does not depend on the trees or tubings and the second of these factors does not depend on the Mellin transform of the primitive, so in this example these different contributions separate into different factors at each order.

The number of leaf tubings of a fixed rooted tree is the number of decreasing labelings (or equivalently the number of increasing labelings) of that tree, see Lemma 3.1 of \cite{BCEFNOHY}.  Increasing labelings of a tree are well-studied objects in combinatorics \cite{bergeron_varieties_1992}.  However, as we run over all trees, leaf tubings are even easier to count.  Given all leaf tubings of all plane rooted trees of size $n$, we can obtain all leaf tubings of all rooted trees of size $n+1$ by inserting a single new leaf in its own tube (and the corresponding new overall tube) in all possible locations in each tubing of each tree.  A rooted tree on $n$ vertices has $2n-1$ insertion places, so inductively, we obtain that there are $(2n-1)!!$ leaf tubings of plane rooted trees of size $n+1$.  This is the same as the number of rooted chord diagrams with $n$ chords.

\subsection{Interpretations for other combinatorial Mellin transforms}

In terms of physically interesting examples, it would be nice to have fully combinatorial interpretations for cases beyond Mellin transforms which are reciprocals of polynomials with all roots nonnegative and rational.  There is a bit that we can do in this direction.

With negative rational roots, one can clear denominators as before and only need to consider additionally the presence of factors of the form $(1-b_ix)$ in the denominator, with $b_i<0$. Taking an alphabet $\Omega_i'$ with $|b_i|$ letters, we can interpret as before, but now there is an overall sign that counts the parity of the number of letters from primed alphabets in the entire compatible labeling.  This gives an interpretation, but note that it is blind to any cancellations of terms with opposite sign since this interpretation doesn't see the signs as signs per se, but merely as another counting variable keeping track of the number of letters from the primed alphabets.

Finite sums of Mellin transforms are also interpretable since the Dyson-Schwinger equation is linear in the Mellin transforms, so we simply interpret each summand and add.  Again if any of these interpretations include signs, we will remain blind to cancellations.

With these observations in mind we can tackle any rational functions with rational poles.

\medskip

In a slightly different direction, if the Mellin transform is the generating series for some combinatorial class $\mathcal{C}$, then by the same arguments used for the form \eqref{eq nice mellin}, we can interpret the perturbative solution for the anomalous dimension as counting tubings of plane rooted trees where the tubes with a common root are given the structure of a $\mathcal{C}$-object.  This is the same notion of putting a $\mathcal{C}$ structure on something that occurs in the interpretation of composition of generating series \cite[Part.~A]{FS}.  %

\section{Asymptotic resurgent analysis}
\label{sec:asymptotics}

For applications in both physics and combinatorics, it is interesting to know the asymptotic growth rate of the coefficients $g_n$ of the power series $\gamma(x)$ in \eqref{eq:gn}. The previously studied examples \cite{bkerfc,bk30,BD,BDM} showed that the coefficients $g_n$ often grow factorially in magnitude and exhibit characteristic resurgent asymptotic behavior. In this section, we prove constraints on this asymptotic behavior directly from the combinatorial tree-tubing result in Theorem~\ref{thm comb de}. We do this using a simplified version of \emph{alien calculus} in the theory of resurgence \cite{ecalle,MR3526111} centered around \emph{the ring of factorially divergent power series} \cite{Bgenfun}. Afterwards, we show that the obtained results are consistent with a transseries ansatz, a standard technique to study resurgence behavior \cite{costinDMJ}. Moreover, we will show that the transseries approach uncovers new behavior if the conditions for the ring of factorially divergent power series are not fulfilled.

\subsection{The ring of factorially divergent power series}
\label{ring}

We can use the resurgence-inspired combinatorial framework from \cite{Bgenfun} to get information on the asymptotic behavior of the coefficients $g_n$ defined in \eqref{eq:gn}. However, the results of \cite{Bgenfun} are not effective at proving specific asymptotic behaviors of solutions to differential equations. Additional assumptions have to be put in, as explained below. 

We briefly describe this framework from \cite{Bgenfun}.
Recall the asymptotic expansion notation that was introduced in Example~\ref{eg de}.
For given constants  $\alpha \in \R_{> 0}$, $\beta \in \R$,
a power series $f(x) = \sum_{n \geq 0} f_n x^n \in \R[[x]]$ is $(\alpha,\beta)$-\emph{factorially divergent} if there is a sequence $c_0^f,c_1^f,\ldots \in \R$ such that
\begin{align} \label{eq:asy} f_n &\,\sim\, \sum_{k\geq 0} \Gamma^{\alpha}_\beta(n-k) \,c_k^f \quad \text{ as } \quad n \rightarrow \infty\,, \end{align}
where $\Gamma^\alpha_\beta(z)= \alpha^{z+\beta} \Gamma(z+\beta)$. Such power series are closed under linear operations, multiplication, division, composition, and differentiation \cite{Bgenfun}. Note that power series whose coefficients $f_n$ grow less than $\Gamma^{\alpha}_\beta(n - R)$ for all $R \geq 0$ are always $(\alpha,\beta)$-factorially divergent. In this case, we can use the trivial sequence $c_k^f = 0$ for all $k \geq 0$. We define $\R^\alpha_\beta[[x]] \subset \R[[x]]$ to be the space of all $(\alpha,\beta)$-factorially divergent power series.
Moreover, we can define a linear map, $\A^\alpha_\beta: \R^\alpha_\beta[[x]] \rightarrow \R[[x]]$, that maps a power series $f \in \R^\alpha_\beta[[x]]$ to the power series $(\A^\alpha_\beta f)(x) =\sum_{k \geq 0} c_k^f x^k$, which encodes the coefficients in the asymptotic expansion \eqref{eq:asy} of $f_n$. This linear map is a derivation that fulfills product and chain rules (see loc.\ cit.\ for details).

Essential for us are the following two properties:
\begin{itemize}
\item[a] (\cite[Proposition 22]{Bgenfun})
If $f,g \in \R^\alpha_\beta[[x]]$ and $h(x) = f(x) g(x)$, then $h \in \R^\alpha_\beta[[x]]$ and 
$$(\A^\alpha_\beta h)(x) = (\A^\alpha_\beta f)(x) g(x) + f(x) (\A^\alpha_\beta g)(x).$$
\item[b]
(\cite[Proposition 15 and 38]{Bgenfun})
If $f \in \R^\alpha_\beta[[x]]$, then $\frac{\partial f}{\partial x} \in \frac{1}{x^2} \R^\alpha_\beta[[x]]$ and 
$$\left(\A^\alpha_\beta \frac{\partial f}{\partial x} \right)(x) = \left( \frac{\alpha^{-1}}{x^2} - \frac{\beta}{x} + \frac{\partial}{\partial x} \right) (\A^\alpha_\beta f)(x).$$
\end{itemize}
From the first property, it also follows that $\R^\alpha_\beta[[x]]$ is a linear subspace of $\R[[x]]$ that is closed under multiplication with elements of $\R[x]$.

Empirically, we can observe that power series solutions of the differential equation system in Theorem~\ref{thm comb de} are often elements of $\R^\alpha_\beta[[x]]$ for some pair $\alpha,\beta$ \cite{BD,BDM,BB}. Unfortunately, \cite[Proposition~38]{Bgenfun} does not guarantee that a specific differential equation system has a solution in $\R^\alpha_\beta[[x]]$ (see also the discussion around \cite[Proposition 55]{Bgenfun}), but only provides statements on this solution under the assumption that it is an element of $\R^\alpha_\beta[[x]]$ for some $\alpha,\beta$. Only for the Yukawa Example~\ref{eg yukawa} is it relatively straightforward to show rigorously that the power series solution of the differential equation is an element of $\R^\alpha_\beta[[x]]$ for specific values of $\alpha, \beta$ (see \cite[Section 7.1]{Bgenfun}). 

Within the context of Theorem~\ref{thm comb de}, we can abbreviate $\gamma_j(x) = x + \sum_{k=1}^j L_k(x)$ to get the differential equation system
\begin{align} \label{ode} \gamma_j(x) \,-\, \gamma_{j-1}(x) \,= \,\gamma(x) \left( 2 x \frac{\partial}{\partial x} \,-\,1 \right) \gamma_j(x) \text{ for } 1 \leq j \leq m\,, \end{align}
where we agree that $\gamma_0(x) = 0$ and we have by definition $\gamma_m(x) = \gamma(x)$.
This system has a \emph{unique} solution in $\R[[x]]$. This solution starts with $\gamma_j(x) = x + \sum_{k=1}^j a_k x^2 + \ldots$
Under the assumption  that $L_k \in \R^\alpha_\beta[[x]]$ or equivalently $\gamma_j \in \R^\alpha_\beta[[x]]$ for some $\alpha,\beta$, we may apply the $\A^\alpha_\beta$-operator to both sides of \eqref{ode}.
We get, after applying the product rule for the derivative $\A^\alpha_\beta$ (a) and the compatibility rule for $\A^\alpha_\beta$ and the ordinary derivative (b), the \emph{linear} differential equation system for the power series $(\A^\alpha_\beta \gamma_1)(x), \ldots, (\A^\alpha_\beta \gamma_m)(x)$,
\begin{gather} \begin{gathered} \label{Aode} (\A^\alpha_\beta \gamma_j)(x) \,-\, (\A^\alpha_\beta \gamma_{j-1})(x) \,= \\
a_j (\A^\alpha_\beta \gamma_m) (x) \left( 2 x \frac{\partial}{\partial x} \,-\,1 \right) \gamma_j(x) \\\,+\, a_j \gamma_m(x) \left( 2 x \frac{\partial}{\partial x} \,-\,1 \,+\, 2 \frac{\alpha^{-1}}{x} \,-\, 2 \beta \right) (\A^\alpha_\beta \gamma_j)(x) \\
\text{ for } 1 \leq j \leq m\,. \end{gathered} \end{gather}
The solution to this system encodes the asymptotic expansions of the power series $\gamma_j$ of the form in eq.~\eqref{eq:asy}.
A simple argument involving the structure of the asymptotic expansion~\eqref{eq:asy}
shows that if and only if $f \in \R^\alpha_\beta[[x]]$ and $\A^\alpha_\beta f \in x^k \R[[x]]$, then $f \in \R^\alpha_{\beta-k}[[x]]$
and $x^k(\A^\alpha_{\beta-k} f)(x) = (\A^\alpha_{\beta} f)(x)$.
\cite[Proposition 14]{Bgenfun}.
So, assuming that $\gamma_j \in \R^\alpha_\beta[[x]]$ for some $\alpha,\beta$ and that the differential system \eqref{ode} is fulfilled 
implies that the linear system \eqref{Aode} holds with  
$(\A^\alpha_\beta \gamma_j)(x) = c_j + d_j x + \bigO(x^2)$ and some specific values for $\alpha,\beta,c_j$ and $d_j$.

Expanding the system~\eqref{Aode} up to order $x$ gives the following constraints on these values:
\begin{align*} 0 &\,=\, c_j - c_{j-1} - 2 a_j \alpha^{-1} c_j \\
0 &\,=\, d_j - d_{j-1} - 2 a_j \alpha^{-1} d_j - a_j \left( c_m - (2 \beta+1) c_j + 2 \alpha^{-1} c_j \sum_{k=1}^m a_k\right) \text{ for } 1 \leq j \leq m\,. \end{align*}
The set of equations in the first line encodes a linear equation system that only has a nontrivial solution for $c_j$ if $\alpha = 2 a_k$ for some $1 \leq k \leq m$. 
From now on for the rest of this section, we require that all $a_j$ be distinct and non-vanishing. Hence, for any $1\leq q \leq m$, we have a 1-parameter family of solutions of the first set of equations,
$\alpha = 2 a_q$, $c_k =0$ for $k < q$, 
$c_k = S_q \prod_{j=q+1}^k\frac{\alpha}{\alpha-2 a_j}$ for $q < k \leq m$, and $c_q = S_q$ an arbitrary value different from $0$.

The second set is an inhomogeneous linear equation system in $d_1, \ldots, d_m$. 
In conjunction with the first set of equations, we may solve directly for $d_1,\ldots,d_{q-1}$,
\begin{align*}      d_k = S_q \left( \prod_{j=q+1}^m\frac{\alpha}{\alpha-2 a_j} \right) \sum_{j=1}^k a_j \prod_{i=j}^{k} \frac{\alpha}{\alpha-2a_i} \text{ for } 1 \leq k \leq q-1\,. \end{align*}
Note that the second system in $d_1,\ldots,d_m$ does not have full rank if we fix $\alpha = 2a_q$.
For the system to have a solution, 
the following equation needs to be fulfilled,
\begin{align} \label{beta_eq} 0 &= - d_{q-1} - a_q \left( c_m - (2 \beta+1) c_q + 2 \alpha^{-1} c_q \sum_{k=1}^m a_k\right). \end{align}

We summarize this discussion of the asymptotic growth rate of the coefficients of the tree-tubing Dyson--Schwinger equations in 
Theorem~\ref{thm comb de} in the following theorem.
\begin{thm}
\label{thm fixed alpha beta}
If all the $a_1,\ldots,a_m$ are different and non-vanishing and the unique set of solution power series $\gamma_1,\ldots,\gamma_m \in \R[[x]]$ of the equation system in eq.~\eqref{ode} are elements of $\R^\alpha_\beta[[x]]$ for some $\alpha,\beta \in \R$, then there is an integer $q$, $1 \leq q \leq m$ such that $(\alpha,\beta) = (\alpha_q,\beta_q)$, where 
\begin{align} \label{eq:alphas} \alpha_q &\,=\, 2 a_q \quad\text{ and }\\
\label{eq:betas} \beta_q &\,=\, \frac{1}{2 a_q}\sum_{\substack{i=1\\ i\neq q}}^m a_i \,+\, \frac12 \prod_{\substack{i=1\\ i\neq q}}^{m}\frac{a_q}{a_q-a_i}\,. \end{align}
Here, the empty sum is $0$, and the empty product is $1$, as usual.
\end{thm}
\begin{proof}
Substitution of the values for $d_{r-1}$ and $c_m$ 
into Eq.~\eqref{beta_eq}  results in 
\begin{align*} \beta &= \frac{1}{2 a_q} \left( -a_q + \sum_{j=1}^q a_j \prod_{\substack{i=j\\ i\neq q}}^{m} \frac{a_q}{a_q-a_i} + \sum_{i=1}^m a_i \right). \end{align*}
It remains to be proven that 
\[
     \sum_{j=1}^q a_j \left( \prod_{\substack{i=j\\ i\neq q}}^{m}\frac{a_q}{a_q-a_i} \right) = a_q \left( \prod_{\substack{i=1\\ i\neq q}}^{m}\frac{a_q}{a_q-a_i} \right),
\]
which is equivalent to 
    \[
    \sum_{j=1}^q a_j \prod_{i=1}^{j-1}\frac{a_q-a_i}{a_q} = a_q.
    \]
We prove this equality by breaking the sum into two pieces
    \[
    \sum_{j=1}^q a_j \prod_{i=1}^{j-1}\frac{a_q-a_i}{a_q} = \underbrace{\sum_{j=1}^{q-1} a_j \prod_{i=1}^{j-1}\left(1 - \frac{a_i}{a_q}\right)}_A + \underbrace{a_q \prod_{i=1}^{q-1}\left(1 - \frac{a_i}{a_q}\right)}_B.
    \]
    Every term in $A$ has the form $(-1)^ta_j/a_q^t$ times a product of $t$ $a_i$s with $i<j$ for some $t\geq 0$. Here, $t$ is the number of factors of the product from which the $a_i/a_q$ was taken to obtain this term.  Furthermore, every such term appears exactly once in $A$.  For $B$, the term where $1$ is taken from each factor when expanding the product is $a_q$.  For every other term in $B$, the explicit $a_q$ can be canceled, so every such term has the form $(-1)^{t+1}/a_q^t$ times a product of $t+1$ $a_i$s for some $t\geq 0$, where here $t+1$ is the number of times the $a_i/a_q$ was taken when expanding out the product, and again every such term appears exactly once. Organizing these terms by the largest index $j$ appearing in the product of $a_i$s we see these terms are exactly the terms obtained from $A$ but with opposite signs.  Therefore the terms of $A$ cancel with all terms of $B$ except for the term $a_q$, giving the desired result.
\end{proof}

\begin{cor}
\label{cor:asymp}
If all the $a_1,\ldots,a_m$ are different and non-vanishing and 
the $L_k$ from Theorem~\ref{thm comb de} are elements of $\R^\alpha_\beta[[x]]$ for some $\alpha,\beta\in \R$,
then the coefficients of 
$$\gamma(x)= \sum_{n \geq 1} g_n x^n = x + \sum_{j=1}^m L_j(x)$$
 have an asymptotic expansion that starts as
\begin{align*} g_n \,\sim\, S_q\, \alpha_q^{n+\beta_q}\, \Gamma(n+\beta_q)\left( 1+ \frac{c_1}{n} + \frac{c_2}{n^2} + \ldots\right) \text{ as } n \rightarrow \infty, \end{align*}
where $q$ is some integer fulfilling $1\leq q \leq m$, and 
$c_1,c_2,\ldots$ are explicitly computable constants,
while $S_q$ is some number in $\R$. 
\end{cor}
\begin{proof}
The statement follows from Theorem~\ref{thm fixed alpha beta} and the defining Eq.~\eqref{eq:asy} of factorially divergent power series.
\end{proof}

Unfortunately, Theorem~\ref{thm fixed alpha beta} must require that the power series solution to the differential equation system \eqref{ode} is an element of $\R^\alpha_\beta[[x]]$ for some pair $\alpha,\beta$.
This requirement is fulfilled in the case of Example~\ref{eg yukawa}.
It would be desirable to have an effective condition for this requirement to hold.

Moreover, it is very reasonable to conjecture that the asymptotic behavior of the coefficients $g_n$ in
the statement above is completely described by the case $q=m$, where $\alpha=2a_m$, which is the 
maximal value that $\alpha$ can attain as the coefficients $a_i$ are required to be ordered. However, it does not necessarily follow from the argument above 
that the prefactor $S$, \emph{the Stokes constant}, is non-zero in this particular case.

An interesting exception is the $m=1$ case, where the prefactor $S$ is known explicitly \cite[\S 7.1]{Bgenfun}.

\subsection{Transseries approach}
\label{sec trans}

The asymptotic analysis in the previous section used the recursive structure of the system of equations in Theorem~\ref{thm comb de}. In this section, we confirm these results using a different method, based instead on the single nonlinear ODE \eqref{eq:ode} in Corollary~\ref{cor:nonlinear}. The solution $\gamma(x)$ in \eqref{eq:gn} is not the most general solution to \eqref{eq:ode} because the solution should depend on $m$ parameters to accommodate the integration constants. Given that the expansion coefficients $g_n$ in \eqref{eq:gn} are often factorially divergent, we expect the general solution to \eqref{eq:ode} to be a {\it transseries} by standard resurgence arguments. (See~\cite{costinDMJ} for precise conditions when the transseries methodology can be applied to systems of nonlinear differential equations.) This transseries is  obtained by extending the formal power series expansion $\gamma$, which we, from now on, denote as $\gamma_{\rm formal}$, in \eqref{eq:gn} to include ``non-perturbative'' terms, which are ``beyond all orders'' in the small $x$ expansion \cite{costinDMJ}. A simple way to generate these terms is to substitute the ansatz
\[
    \gamma(x)=\gamma_{\rm formal}(x)+\sigma\, \gamma_{\rm trans}(x)
\]
into \eqref{eq:ode} and linearize in the parameter $\sigma$, meaning that we expand in $\sigma$ and drop all nonlinear terms. This yields a homogeneous {\it linear} ODE, of order $m$, for $\gamma_{\rm trans}(x)$, in which the original formal series $\gamma_{\rm formal}(x)$ appears in the coefficients. Motivated by the general procedure in \cite{costinDMJ} we assume that the non-perturbative contributions are of the form %
\begin{eqnarray*} \gamma_{\rm trans}(x)= \frac{1}{x^\beta} \, \exp\left[-\frac{1}{\alpha\, x}\right]\left(1+\bigO(x)\right) \text{ as } x\to 0^+. \end{eqnarray*}
Here, $\bigO(x)$ stands for some power series in $x\,\Q[[x]]$.
Substitution of this ansatz into the linearized equation for $\gamma_{\rm trans}(x)$ determines $m$ possible values for the pair of parameters $(\alpha, \beta)$. For the ansatz to yield a consistent solution of \eqref{eq:ode}, $\alpha$ has to be a root of a specific polynomial. The nonlinear ODE \eqref{eq:ode} has a special factorized form, and this has the consequence that this polynomial factorizes into $m$ linear factors. We find that the possible $\alpha$ parameters are in the set $\{\alpha_1,\ldots,\alpha_m\}$ with $\alpha_q$ as defined in Theorem~\ref{thm fixed alpha beta}.

For each such value of $\alpha$ the next term in the small $x$ expansion of the linearized equation for $\gamma_{\rm trans}(x)$ determines the parameter $\beta$ to be %
equal to the corresponding value $\beta_q$ in Theorem~\ref{thm fixed alpha beta}.
We can only solve the associated equation for $\beta$ if 
all the $a_j$ are distinct. That means if we assume that $a_1<a_2< \dots <a_m$, then there are $m$ possible forms of the non-perturbative contributions given by %
\begin{eqnarray} \label{eq:trans} \gamma_{\rm trans}^{[q]}(x) = \frac{1}{x^{\beta_q}}\, \exp\left[-\frac{1}{\alpha_q\, x}\right] \left(1+\bigO(x)\right) \text{ for } q\in \{1, \dots, m\}\,, \end{eqnarray}
where the pair of parameters $(\alpha_q, \beta_q)$ are given in \eqref{eq:alphas} and \eqref{eq:betas}.
Thus, we may see the full formal solution space to \eqref{eq:ode} as given by the transseries
$$
    \gamma(x)=\gamma_{\rm formal}(x)+\sum_{q=1}^m \sigma_q\, \gamma_{\rm trans}^{[q]}(x) + \bigO(\sigma^2)\,,
$$
where higher-order contributions in $\sigma_1,\ldots,\sigma_m$ can be computed mechanically
(see, e.g., \cite{BD,BDM,BB}).

\begin{remark}
    The correspondence between the parameters $(\alpha_p, \beta_p)$ in \eqref{eq:trans} and in Theorem~\ref{thm fixed alpha beta} is an explicit example of resurgence. We find a quantitative relationship between the $n\to \infty$ asymptotics of the coefficients $g_n$ of the formal solution $\gamma_{\rm formal}(x)$ and the non-perturbative small $x$ corrections $\gamma_{\rm trans}(x)$ (see \cite{costinDMJ} for a detailed treatment of this correspondence). 

In the dictionary that translates between the transseries and the asymptotic information,
the integration constants $\sigma_1,\ldots,\sigma_m$ correspond to the Stokes constants of the asymptotic expansions.
\end{remark}

\begin{example}[First order] 
For $m=1$ the formal solution begins as follows:
\begin{eqnarray*} \gamma_{\rm formal}(x)= x +a_1 x^2 +4a_1^2 x^3+27 a_1^3 x^4+\dots \end{eqnarray*}
The associated non-perturbative contribution is of the form
\begin{eqnarray*} \gamma_{\rm trans}(x) = x^{-\frac{1}{2}}\, \exp\left[-\frac{1}{2\,a_1\, x}\right] \left(1+\bigO(x)\right) \,. \end{eqnarray*}
\end{example}

\begin{example}[Second order with different values of $a_1$ and $a_2$]
\label{sec:degeneracy}
For $m=2$, with $a_1<a_2$, the formal solution begins as follows:
\begin{eqnarray*} \gamma_{\rm formal}(x)&=& x +(a_1+a_2) x^2 +(4(a_1^2+a_2^2)+5 a_1 a_2)) x^3 \\
 &&+(27 (a_1^3+a_2^3)+40(a_1^2 a_2+a_1 a_2^2)) x^4+\dots \end{eqnarray*}
The two independent non-perturbative contributions are of the form
\begin{eqnarray*} \gamma_{\rm trans}^{[1]}(x) &=& x^{-\frac{a_2}{2a_1}-\frac{a_1}{2(a_1-a_2)}}\, \exp\left[-\frac{1}{2\,a_1\, x}\right]\left(1+\bigO(x)\right) \\
 \gamma_{\rm trans}^{[2]}(x) &=& x^{-\frac{a_1}{2a_2}-\frac{a_2}{2(a_2-a_1)}}\, \exp\left[-\frac{1}{2\,a_2\, x}\right]\left(1+\bigO(x)\right) \,. \end{eqnarray*}
With $\{a_1, a_2\}=\{1,2\}$, as in \cite{BB}, the formal perturbative solution is $\gamma_{\rm formal}(x)= x+3x^2 +30x^3+483 x^4+ \dots$, and the non-perturbative contributions are of the form
\begin{eqnarray*} \gamma_{\rm trans}^{[1]}(x) = x^{-\frac{1}{2}}\, \exp\left[-\frac{1}{2\, x}\right] \left(1+\bigO(x)\right) \,, \quad \gamma_{\rm trans}^{[2]}(x) =x^{-\frac{5}{4}}\, \exp\left[-\frac{1}{4\, x}\right]\left(1+\bigO(x)\right)\,. \end{eqnarray*}
\end{example}
\begin{example}[Second order with $a_1=a_2$]
If $a_1=a_2$, then the situation is different.
As the numbers $a_i$ are not strictly ordered anymore, we cannot apply Theorem~\ref{thm fixed alpha beta}.
However, we can still go forward with a modified transseries ansatz.
We fix $a_1=a_2=1$ and find that the formal power series solution of \eqref{eq:ode} starts with
\begin{eqnarray*} \gamma_{\rm formal}(x)= x+2x^2+13x^3+134x^4+ 1811x^5+29568x^6 +\dots \end{eqnarray*}
Experimentally, we find that the procedure described at the beginning of Section~\ref{sec trans}
still works if we modify our transseries ansatz to have the following shape:\footnote{The precise shape of this ansatz can be deduced from a Frobenius analysis of the second-order linear homogeneous ODE satisfied by $\gamma_{\rm trans}$.}
\begin{align*} \gamma(x) &= \gamma_{\rm formal}(x) + \sigma_1 \gamma_{\rm trans}^{[1]}(x) + \sigma_2 \gamma_{\rm trans}^{[2]}(x) + \bigO(\sigma^2)\,, \intertext{where} \gamma_{\rm trans}^{[1]}(x)& = \frac{1}{x}\, e^{-\frac{1}{2x}+\frac{1}{\sqrt{x}}}(1+\bigO(\sqrt{x})) \,, \quad \gamma_{\rm trans}^{[2]}(x) =\frac{1}{x}\, e^{-\frac{1}{2x}-\frac{1}{\sqrt{x}}}(1+\bigO(\sqrt{x})). \end{align*}
With this ansatz, we may proceed as before.
Evidently, the exponential terms have a completely different form when $a_1\neq a_2$ as in the previous example.
Here, we see the appearance of two different non-perturbative scales, $\exp\left[-{1}/{\sqrt{x}}\,\right]$ and 
$\exp\left[-{1}/{2x}\right]$, where the first of the scales has two different phases $\exp\left[-{1}/{\sqrt{x}}\,\right]$ and $\exp\left[{1}/{\sqrt{x}}\,\right]$. 
However, we have little to no information on the asymptotic behavior of the coefficients $g_n$ in this degenerate case. 
The sequence $g_n$ is not $(\alpha,\beta)$-factorially divergent in the sense of the definitions in Section~\ref{ring}.
The available tools from \cite{Bgenfun} and \cite{costinDMJ} do not seem applicable here.  A promising tool to study this asymptotic behavior is \'Ecalle acceleration \cite{ecalle,costin-ecalleacc}, which allows the asymptotic treatment of transseries with different scales.
\end{example}

We do not know if such a degenerate case, where the values of two $a_i$ coincide and $F(\rho)$ has a pole of order two, is 
allowed by the various constraints that quantum field theory imposes on $F(\rho)$.
However, we can provide an interpretation of this degeneracy and the different 
asymptotics it produces on the combinatorial side:

\begin{remark}[Combinatorial heuristics for the degenerate $a_1=a_2$ case]
Experimentally, working with some thousands of terms
generated with the \texttt{gamma.py} program (see Remark~\ref{remark gamma program}), it appears that the large-order asymptotics of the perturbative expansion coefficients changes abruptly when $a_1=a_2$; another reflection of the fundamental difference between the behavior of the solution when $a_1=a_2$ vs $a_1 \neq a_2$. 

We can give a combinatorial heuristic for why the $a_1=a_2$ case ought to have different growth behavior. However, note that this could only be part of a rigorous argument in the special case of the leaf tubings, see Subsection~\ref{sec other intermediate approx}.  The heuristic is as follows: any difference in the behavior depending on properties of the values of the $a_i$ should be seen combinatorially in the compatible labelings. The effect of the compatible labeling constraint is most pronounced when many tubes share a root since that is when the admissibility constraint on words comes into play most strongly.  Consequently, for a heuristic as to why $a_1=a_2$ should behave differently from $a_1\neq a_2$ we can look to why the number of admissible words behaves differently when $a_1=a_2$ compared to $a_1\neq a_2$.  For the leaf tubings, all tubes of size $>1$ share a root, and so counting the number of compatible labelings of a leaf tubing of a tree on $n+1$ vertices is exactly counting the number of admissible words of length $n$.  For other tubings, we do not have an exact count for the number of compatible labelings, but we can still use the number of admissible words as a useful heuristic for what kinds of differences we might see.

With this in mind, let us enumerate admissible words of length $n$ with $m=2$ alphabets.  This is amenable to direct enumeration: simply sum over the possible lengths of the word on alphabet $\Omega_1$ to obtain $\sum_{k=0}^{n}a_1^k a_2^{n-k}$.  Now consider the asymptotic behavior of this expression for large $n$. In the case $a_1=a_2 =a$,
\[
\sum_{k=0}^{n}a_1^k a_2^{n-k} = \sum_{k=0}^{n}a^{n} = (n+1)a^{n}, 
\]
while in the case $a_1<a_2$, 
\[
\sum_{k=0}^{n}a_1^{k}a_2^{n-k}  = a_2^{n}\sum_{k=0}^{n}\left(\frac{a_1}{a_2}\right)^{n-k} = a_2^{n}\left(\frac{1-\left(\frac{a_1}{a_2}\right)^{n+1}}{1-\frac{a_1}{a_2}}\right).
\]
So, while both have exponential growth with base $a_2$, when $a_1=a_2$, there is a subexponential factor of $n$, which does not appear in the $a_1\neq a_2$ case.
\end{remark}

\begin{example}[Third-order case] For $m=3$, with $a_1<a_2<a_3$, the formal solution begins as follows:
\begin{eqnarray} \gamma_{\rm formal}(x)&=& x +(a_1+a_2+a_3) x^2 +(4(a_1^2+a_2^2+a_3^2)+5 (a_1 a_2+a_1 a_3+a_2 a_3)) x^3 \nonumber\\
 &&\hskip -2.5cm +(27(a_1^3+a_2^3+a_3^3)+40(a_1^2 a_2+a_1^2 a_3+a_2^2 a_1+a_2^2 a_3 +a_3^2 a_1+a_3^2 a_2)+54 a_1 a_2 a_3) x^4 +\dots \nonumber \end{eqnarray}
The three independent  non-perturbative contributions are of the shape
\begin{align*} \gamma_{\rm trans}^{[1]}(x) &= x^{-\frac{a_2+a_3}{2a_1}-\frac{a_1^2}{2(a_1-a_2)(a_1-a_3)}}\, \exp\left[-\frac{1}{2\,a_1\, x}\right] \left(1+\bigO(x)\right) \\
\gamma_{\rm trans}^{[2]}(x) &= x^{-\frac{a_1+a_3}{2a_2}-\frac{a_2^2}{2(a_2-a_1)(a_2-a_3)}}\, \exp\left[-\frac{1}{2\,a_2\, x}\right] \left(1+\bigO(x)\right) \\
 \gamma_{\rm trans}^{[3]}(x) &= x^{-\frac{a_1+a_2}{2a_3}-\frac{a_3^2}{2(a_3-a_1)(a_3-a_2)}}\, \exp\left[-\frac{1}{2\,a_3\, x}\right] \left(1+\bigO(x)\right)\,. \end{align*}
As for the $m=2$ case, we see that when any of the $a_j$ parameters are equal, the transseries structure changes abruptly, as some of the exponents above will become ill-defined. %

For the $\phi^3$ example, we have $\{a_1, a_2, a_3\}=\{2, 3, 6\}$, and the formal solution is: $\gamma_{\rm formal}(x)= x+11x^2 +376x^3+20241 x^4+ \dots$, 
\cite[\href{https://oeis.org/A051862}{A051862}]{oeis}. The 3 independent non-perturbative contributions are of the shape \cite{BDM}
\begin{align*} \gamma_{\rm trans}^{[1]}(x) &= x^{-\frac{11}{4}} \exp\left[-\frac{1}{4\, x}\right] \left(1+\bigO(x)\right) \\
\gamma_{\rm trans}^{[2]}(x) &= x^{+\frac{1}{6}} \exp\left[-\frac{1}{6\, x}\right] \left(1+\bigO(x)\right) \\
 \gamma_{\rm trans}^{[3]}(x) &= x^{-\frac{23}{12}} \exp\left[-\frac{1}{12\, x}\right] \left(1+\bigO(x)\right). \end{align*}
\end{example}
An interesting feature of this $\phi^3$ model is the appearance of logarithmic terms \cite{BDM,BB} at higher orders in the transseries for the anomalous dimension $\gamma(x)$, due to the resonance effect: $\{\frac{1}{a_1}, \frac{1}{a_2}, \frac{1}{a_3}\}=\{\frac{1}{2}, \frac{1}{3}, \frac{1}{6}\}=\{3, 2, 1\}\times \frac{1}{6}$. We leave to future work the interesting combinatorial question of investigating the general conditions for such logarithmic terms to be generated.

\end{document}